\documentclass[preprint,3p]{elsarticle}

\newtheorem{theorem}{Theorem}
\newdefinition{remark}[theorem]{Remark}
\newproof{proof}{\textit{Proof}}
\newdefinition{definition}[theorem]{Definition}
\newdefinition{example}[theorem]{Example}

\usepackage{amssymb}
\usepackage[charter,cal=cmcal]{mathdesign}
\usepackage[utf8]{inputenc}
\usepackage[T1]{fontenc}
\usepackage[colorlinks]{hyperref}
\usepackage{url}
\usepackage{microtype}

\usepackage{amsmath}
\usepackage[small,bold,langfont=roman]{complexity}
\newclass{\SigmaTwoP}{\boldsymbol{\Sigma}_2^\P}
\newclass{\PiTwoP}{\boldsymbol{\Pi}_2^\P}
\newcommand{\N}{\mathbb{N}}
\newcommand{\set}[1]{\{#1\}}
\newcommand{\parts}[1]{2^{#1}}
\newcommand{\rsA}{\mathcal{A}}

\newcommand{\rsM}{\mathcal{M}}

\newcommand{\react}[3]{(#1,#2,#3)}
\newcommand{\restart}{\triangleleft}
\DeclareMathOperator{\res}{res}
\DeclareMathOperator{\posvar}{pos}
\DeclareMathOperator{\negvar}{neg}
\newcommand{\blank}{\Box}

\usepackage{xspace, comment}
\newcommand{\RS}{RS\xspace}

\newcommand{\ie}{\emph{i.e.},\xspace}
\newcommand{\etc}{\emph{etc.}\@\xspace}

\begin{document}

\title{Complexity of the dynamics of reaction systems\tnoteref{blabla}
}
\tnotetext[blabla]{This paper collects, in a more extended/generalized version, both results presented at CiE 2014 conference~\cite{Formenti2014a} and the ones presented at DCFS 2014 conference~\cite{Formenti2014b}. It also contains a considerable number of further results.}
\author[unimib]{Alberto Dennunzio}
\ead{dennunzio@disco.unimib.it}

\author[unice]{Enrico Formenti}
\ead{enrico.formenti@unice.fr}

\author[unimib]{Luca Manzoni}
\ead{luca.manzoni@disco.unimib.it}

\author[unimib]{Antonio E. Porreca}
\ead{porreca@disco.unimib.it}

\address[unimib]{Dipartimento di Informatica, Sistemistica e Comunicazione\\
  Università degli Studi di Milano-Bicocca\\
  Viale Sarca 336/14, 20126 Milano, Italy}

\address[unice]{Universit\'e C\^ote d'Azur (UCA), CNRS, I3S, France}

\begin{abstract}
  Reaction systems are discrete dynamical systems inspired by bio-chemical processes, whose dynamical behaviour is expressed by 
  set-theoretic operations on finite sets. Reaction systems thus provide a 
  description of bio-chemical phenomena that complements the more traditional approaches, for instance those based on differential equations.
  A comprehensive list of decision problems about the dynamical behavior of reaction systems (such as 
  cycles and fixed/periodic points, attractors, and reachability) 
  is provided along with the corresponding computational complexity, which ranges from tractable problems to $\PSPACE$-complete problems.
\end{abstract}


\maketitle

\section{Introduction}

Reaction systems (\RS) are a computational model recently introduced by Ehrenfeucht and Rozenberg~\cite{Ehrenfeucht2007a} which is inspired by chemical reactions.
Roughly speaking, a RS is made of a finite set of entities (molecules of substances) and a finite set of rules (reactions). Entities are used as reactants, inhibitors, and products of a reaction. If, among a current set of entities (state), the reactants of a given reaction are available and there are no inhibitors, then that reaction takes place, \ie the products of the considered reaction are generated. The collection of the products of all enabled reactions give the new state of the system. In fact, this framework therefore captures the fundamental mechanisms of facilitation and  inhibition on which biochemical reactions are based.
Further, the functioning of the model undergoes two important principles, namely, the threshold assumption and the non-permanency assumption. The first one means that if a resource is available then there is always a sufficient amount of it for all the reactions which can take place. In other terms, reactions needing the same resource do not come into conflict. This principle essentially defines \RS as a qualitative model (\ie  only the presence or absence of a substance is measured). The latter assumption imposes that an entity will disappear from the current state of the system if no enabled reaction will have produced it. 

Thus, \RS are very different from the standard frameworks for biomodelling such as, for instance, ordinary differential equations and continuous stochastic process. However, it has been recently shown that they are able to capture the main quantitative characteristics of the ODE-based model for biochemical processes. In~\cite{Azimi2014b}, for instance, they show how to build a \RS reproducing the quantitative behavior of the ODE model for the molecular heat shock response. Another important aspect in the \RS framework is that, differently than in the traditional models, one is able to follow the cause-effect relations between reactions, allowing for instance to understand what led a reaction to take place. 

As consequence of these facts and the capability of \RS to capture the main mechanisms of biochemical reactions, the interest and use of \RS to study practical problems concerning biological and/or chemical processes~\cite{corolli2012,Azimi2014b}  has grown since their introduction.

The study of the dynamical properties of the model at hand is a central question in the context of modelling. When considering \RS as a model, this question
might seem trivial at a first glance. Indeed, \RS are finite systems and, as such, their dynamics is ultimately periodic. However, in practical applications
this information is pointless and a more detailed analysis is expected. For example, in the context of genetic modelling, the number of attractors can
be associated with cell differentiation~\cite{Shmulevich2002a}, the period of attractors may represent the period of life process (respiration, 
circadian cycle, \etc), the size of the basin of attraction can be associated with the robustness in complex biological systems~\cite{Kitano2004a}.
The analysis can be pushed even further considering limited resources~\cite{Ehrenfeucht2011a}.   

In~\cite{Ehrenfeucht2010a}, they considered some particular state as a death state (the empty set in this case) and computed the probability of a system reaching the death state.

Other studies focused on the complexity of deciding if any given \RS exhibits a certain dynamical behaviour as, for instance, the appearance of a specific product during the evolution~\cite{Salomaa2013a,Salomaa2013b}. 

This paper pursues the seminal results in \cite{Ehrenfeucht2007a,Salomaa2013a,Salomaa2013b,
Formenti2014c,Dennunzio2015a} about
the dynamical behavior of \RS. 
The purpose is to cover a broad range of questions involving cycles, attractors and reachability, providing in this way a reference list for them. 
Several results are stated in a form
which is as wide as possible (including constraints as ``at least'' and ``at most'') in order to answer the different variants/formulations of a same problem which arise in practical situations. 
Remark that similar problems are studied also in the context of other models used for applications, as, for instance, Boolean Automata Networks and Cellular Automata~\cite{sutner1995a, Dennunzio2016a}.

The first group of results concern the complexity of the decision problems about various forms of reachability under the action of a \RS. We provide the complexity of establishing if  
\begin{itemize}
\item
a state $T$ leads to a state $U$ in at most $k$ steps ($\P$-complete, Theorem~\ref{thm:polynomial-time-reachability})
\item
a state $T$ leads to a state $U$ ($\PSPACE$-complete, Theorem~\ref{thm:reachability})
\item
a state $T$ leads to a cycle of length at most $\ell$ after at most $k$ steps ($\P$-complete, Theorem~\ref{thm:reachability-cycle}) or after at least $k$ steps ($\PSPACE$-complete, Theorem~\ref{thm:cycle-with-long-handle})
\item
there exists a state $T$ leading to a cycle of length at most $\ell$ after at least $k$ steps ($\NP$-complete, Theorem~\ref{thm:reachability-cycle-exists}),
\item
a state $T$ leads to a cycle of length at least $\ell$ after at least $k$ steps ($\PSPACE$-complete, Theorem~\ref{thm:cycle-with-long-handle}) or if such a state exists ($\PSPACE$-complete, Theorem~\ref{thm:reachability-large-cycle-exists}), 
\item a state $T$ is reachable from some other state in $k$ steps ($\NP$-complete, Theorem~\ref{thm:k-ancestor}).
\item a state $T$ is periodic of period at most~$\ell$ and reachable from some other ultimately periodic state having a pre-period of length at least~$k$ ($\NP$-complete, Theorem~\ref{thm:T-in-short-cycle-with-long-handle}) 
\end{itemize}
As to problems involving cycles, we deal with the complexity of determining if\begin{itemize}
\item
a state $T$ belongs to a cycle of length at most $\ell$ ($\P$-complete, Theorem~\ref{thm:state-in-short-cycle})
\item there exists a cycle of length at most~$\ell$ ($\NP$-complete in both the versions with ``for each~$\ell$'' in the statement and with $\ell$ given as input, Theorems~\ref{thm:has-short-fixed-cycle} and~\ref{thm:has-short-cycle})
\item there exists a cycle of length at least~$\ell$ ($\PSPACE$-complete, Theorem~\ref{thm:has-large-fixed-cycle})
\item the given \RS consists of cycles possibly all of length 1 ($\coNP$-complete, Theorems~\ref{thm:bijection} and~\ref{thm:identity})
\end{itemize}
Concerning the attractors, we study the complexity of deciding if 
\begin{itemize}
\item a state $T$ belongs to a local attractor cycle of length at most $\ell$ with attraction basin of diameter at most $d$/at  least $d$ ($\coNP$-complete/$\NP$-complete, Theorems~\ref{thm:small-basin-for-T} and~\ref{thm:large-basin-for-T})
\item there exists a local attractor cycle of length at most $\ell$ with attraction basin of diameter at least $d$/at  most $d$ ($\NP$-complete/$\SigmaTwoP$-complete, Theorems~\ref{thm:large-basin-exists} and~\ref{thm:small-basin-exists})
\item a state $T$ belongs to a global attractor cycle of length at most $\ell$/at least $\ell$  ($\PSPACE$-complete, Theorem~\ref{thm:small-global})
\item there exists a global attractor cycle of length at most $\ell$/at least $\ell$  ($\PSPACE$-complete, Theorem~\ref{thm:small-global})
\end{itemize}
Some proofs of our results exploit the idea that \RS can be used to evaluate Boolean formulae~\cite{Ehrenfeucht2007a} and are based on reductions from problems over these latter. 
Other ones consist of reductions from different halting problems over deterministic Turing machines with bounded tape. Indeed, as described in Section~\ref{sec:turing-machines}, the computations of such machines can be simulated by suitable \RS 
Finally, direct reductions from the problems considered in this paper allow to prove the remaining results. 
\medskip

The paper is structured as follows. Section~\ref{sec:basic} provides the basic notions on \RS and Section~\ref{sec:turing-machines} illustrates the construction of a \RS simulating a Turing machine with bounded tape. The results regarding reachability issues, cycles and attractors appear in the three Sections~\ref{sec:easy-problems},~\ref{sec:hard-problems}, and~\ref{sec:very-hard-problems}, depending on the complexity degree of the involved problems.  
Indeed, this allows a certain uniformity of proof techniques and to factor some parts or ideas.
Section~\ref{sec:conclusion} provides a recap of the results and discusses possible future developments.

\section{Basic notions}
\label{sec:basic}

This section provides all the basic concepts about \RS and those notions about dynamical systems which will be considered in the paper. Notations are taken from~\cite{Ehrenfeucht2007a}. 

Let~$S$ be a finite set and let its elements be called \emph{entities}. Denote by $\parts{S}$ the power set of~$S$.

\begin{definition}
A \emph{reaction}~$a$ over~$S$ is a triple~$\react{R_a}{I_a}{P_a}$, where ${R_a}$, ${I_a}$ and ${P_a}$ are subsets of~$S$ called the set of \emph{reactants}, the set of \emph{inhibitors}, and the of \emph{products}, respectively. The collection of all reactions over~$S$ is denoted by~$\mathrm{rac}(S)$. 
\end{definition}

\begin{definition}
  A \emph{reaction system}~$\rsA$ is a pair~$(S,A)$ where~$S$ is a finite set, called the \emph{background set}, and~$A\subseteq\mathrm{rac}(S)$. Every subset of the background set of~$\rsA$ is said to be a \emph{state}.
\end{definition}

Remark that $R$ and $I$ are allowed to be empty as in the original definition of \RS from~\cite{Ehrenfeucht2007a}. Another option, not used here, is to assume the non emptiness of these two sets, as done later in~\cite{Ehrenfeucht2010a,Ehrenfeucht2011a}, for instance. We want to stress that the essence of our results does not change at all in the non emptiness setting, while our choice allows a more elegant formulation of the results.

\smallskip

For any state~$T \subseteq S$ and any reaction~$a\in\mathrm{rac}(S)$, $a$ is said to be \emph{enabled} in~$T$ if~$R_a \subseteq T$ and~$I_a \cap T = \varnothing$. The \emph{result function}~$\res_a \colon \parts{S} \to \parts{S}$ of~$a$ is defined as
\[\forall T\subseteq S, \quad
\res_a(T) =
\begin{cases}
    P_a & \text{if~$a$ is enabled in~$T$}\\
    \varnothing & \text{otherwise.}
\end{cases}
\]
The definition of~$\res_a$ naturally extends to sets of reactions. Indeed, given a set of reaction $B\subseteq\mathrm{rac}(S)$, define~$\res_B\colon \parts{S} \to \parts{S}$ as $\res_{B}(T)=\bigcup_{a \in B} \res_{a}(T)$ for every $T\subseteq S$. The result function~$\res_{\rsA}$ of a \RS~$\rsA = (S,A)$ is~$\res_{A}$. In this way, the discrete dynamical system with set of states~$\parts{S}$ and next state map~$\res_{\rsA}$ is associated with the \RS~$\rsA = (S,A)$.

The following is an example of a simple \RS which is able to simulate a NAND gate.

\begin{example}[NAND gate]
  To implement a NAND gate using a \RS we use  $S = \set{0_a, 1_a, 0_b, 1_b, 0_{\text{out}}, 1_{\text{out}}}$ as background set. The first four elements represent the two inputs (denoted by the subscripts $a$ and $b$), the last two, on the other hand, denote the two possible outputs. The reactions used to model a NAND gate are the followings: $(\set{0_a,0_b},\varnothing,\set{1_{\text{out}}})$, $(\set{0_a,1_b},\varnothing,\set{1_{\text{out}}})$, $(\set{1_a,0_b},\varnothing,\set{1_{\text{out}}})$, and $(\set{1_a,1_b},\varnothing,\set{0_{\text{out}}})$. Similarly to NAND gates, others gates can be simulated and it is possible to build circuits with gates of limited fan-in using only a number of entities and reactions that is linear in the size of the modelled circuit.
\end{example}

We now proceed by recalling the necessary definitions of the dynamical properties investigated in this work.

Let~$\rsA = (S,A)$ be a \RS. For any given state $T \subseteq S$, the \emph{dynamical evolution} or \emph{dynamics} of~$\rsA$ starting from $T$ is  the sequence $(T, \res_{\rsA}(T), \res_{\rsA}^{2}(T), \ldots)$ of states visited by the system starting from $T$, \ie the sequence in which for every $i \in \N$ the $i$-th element is $\res_{\rsA}^{i}(T)$. We say that  \emph{$T$ leads to a state $U$ in $t$ steps} if $U$ belongs to the dynamics starting from $T$ and $U=\res_{\rsA}^{t}(T)$. In that case, $T$ is said to be a \emph{$t$-ancestor} of $U$. We simply say that $T$ leads to $U$ if $T$ leads to a state $U$ in $t$ steps, for some $t\in\N$. Since $\parts{S}$ is finite, any of the above sequences is ultimately periodic, \ie there exist two integers $h\in \N$  and $k>0$ such that $\res_{\rsA}^{h+k}(T)=\res_{\rsA}^{h}(T)$. The smallest values of $h$ and $k$ for which this condition holds are called \emph{preperiod} and \emph{period} of $T$, respectively. If $h=0$, $T$ is a \emph{periodic point} and the set $\{T, \res_{\rsA}(T), \ldots, \res_{\rsA}^{k-1}(T)\}$ is said to be a \emph{cycle} (of \emph{length} $k$). We say that a periodic point $T$ is a \emph{fixed point} if $k=1$.

 The notion of attractor is a central concept in the study of dynamical systems. Recall that an \emph{invariant set} for $\rsA$ is a set of states $\mathcal{U}$ with $\res_{\rsA}(\mathcal{U}) = \set{\res_{\rsA}(U) : U \in \mathcal{U}} = \mathcal{U}$. Remark that each invariant set of any \RS consists of cycles. A \emph{local attractor} in a \RS is an invariant set $\mathcal{U}$ such that there exists $T \notin \mathcal{U}$ with $\res_{\rsA}(T) \in \mathcal{U}$. Intuitively, a local attractor is a set of states $\mathcal{U}$ from which the dynamics never escapes and such that there exists at least one external state whose dynamics ends up in $\mathcal{U}$. A \emph{global attractor} for an \RS $\rsA$ is an invariant set of states $\mathcal{U}$ such that for all $T \in \parts{S}$ there exists $t \in \N$ such that $\res_{\rsA}^{t}(T) \in \mathcal{U}$. A local/global attractor $\mathcal{U}$ is a \emph{local/global fixed-point attractor} if $\mathcal{U} = \set{T}$ and hence, necessarily, $T$ is a fixed point. Similarly, we call a local/global attractor $\mathcal{U}$ a \emph{local/global attractor cycle} if all the states in $\mathcal{U}$ belong to the same cycle. The \emph{attraction basin} of a local/global attractor $\mathcal{U}$ is the set of all states whose dynamics ends up in $\mathcal{U}$. The maximum value of all the preperiods of such states is called the \emph{diameter} of $\mathcal{U}$.

\section{Simulating Turing machines with bounded tape}
\label{sec:turing-machines}
In this section we 
illustrate the simulation of Turing machines with bounded tape performed by {\RS}s. 
For an introduction, basic results, and notions on Turing machines, we refer the reader to~\cite{Hopcroft1979a}.

\smallskip

Let~$M$ be any single-tape deterministic Turing machine with tape of length $m$, set of states~$Q$, tape alphabet~$\Sigma$, and transition function $\delta \colon Q \times \Sigma \to Q \times \Sigma \times \set{-1,0,+1}$. The computations of~$M$ can be simulated by the reaction system~$\rsM = (S,A)$ defined as follows.

\paragraph{Entities}
The set of entities of~$\rsM$ is 
\begin{align*}
  S = \set{a_i : a \in \Sigma, 1 \le i \le m} \cup
  \set{q_i : q \in Q, 1 \le i \le m+1}
\end{align*} 
that is, $S$ consists of all symbols of the alphabet and all states of $M$, each of them indexed by every possible tape position (with the additional position $m+1$ for the states, too).

The generic configuration where $M$ is in state $q\in Q$, its tape head is located on cell $i$, and its tape consists of the string $x=x_1\cdots x_m$, is encoded as the following state
\begin{align*}
  T = \set{(x_j)_j : 1 \le j \le m} \cup
  \set{q_i}
\end{align*} 
In other terms, $T$ contains each symbol of $x$ indexed by its position on the tape and the entity $q_i$ (state entity) storing both the current state and the head position of $M$ (see Fig.~\ref{fig:encodingconf}).

\begin{figure}
\begin{center}
\includegraphics[page=1]{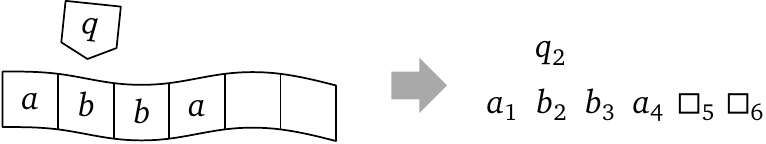}
\end{center}
\caption{Encoding of a configurations of~$M$ as state $T=\set{a_1, b_2, b_3, a_4, \blank_5, \blank_6, q_2}$ of $\rsM$ where each tape symbol is subscripted with its own position, while the current state of $M$ is subscripted with  the position of the head. The square~$\blank$ represents the blank symbol.}
\label{fig:encodingconf}
\end{figure}

\paragraph{Reactions}
We are going to illustrate the elements of the set $A$. Each transition~$\delta(q,a) = (q',a',d)$ of $M$, with $q,q'\in Q$, $a, a'\in\Sigma$, and $d\in\{-1,0,1\}$, gives rise to the following set of reactions:
\begin{align}
  \label{react:transition}
  &(\set{q_i,a_i}, \varnothing, \set{q'_{i+d},a'_i})&
  &\text{for } 1 \le i \le m
\end{align}
If the head of $M$ is located on the cell $i$ of the tape, then the $i$-th reaction simulates the action of $\delta$, by producing the entity encoding both the new state and the new head position of $M$, as well as the symbol written on the tape in position $i$\footnote{Notice that, for simplicity, the background set~$S$ and the reactions of type~\eqref{react:transition} refer to the entities~$q_{m+1}$ for all~$q \in Q$; however, in this section the tape head, by hypothesis, does never reach the~$(m+1)$-th tape cell. The entities~$q_{m+1}$ will play a different role in Section~\ref{sec:very-hard-problems}.}. 


The following reactions assure that the symbols on the tape where the head is not located keep unchanged:
\begin{align}
  \label{react:preserve-tape}
  &(\set{a_i}, \set{q_i : q \in Q}, \set{a_i})&
  &\text{for } 1 \le i \le m
\end{align}
Indeed, for every $i$, the $i$-th reaction of type~\eqref{react:preserve-tape} is inhibited by the presence of any state with the same subscript $i$, to allow that only a reaction of type~\eqref{react:transition} takes place when the head of $M$ is over the tape cell $i$.

\smallskip

Hence, the reaction system~$\rsM$ maps the state encoding any configuration of~$M$ on input~$x$ to the state corresponding to the next configuration of~$M$. Notice that some states of~$\rsM$ are ``malformed'', as they contain zero or more than one state entity, or zero or more than one symbol making reference to a position. In this case, the behaviour of~$\rsM$ is unrelated with the one of~$M$.

We also stress that the reaction system~$\rsM$ can be built from the description of~$M$ and the unary encoding of~$m$ in logarithmic space, since its reactions can be output by iteration over all tape positions~$1 \le i \le m$ in binary notation.

\section{Problems Solvable in Polynomial Time}
\label{sec:easy-problems}

Some interesting problems related to the dynamics of reaction systems involve their forward dynamics and for a limited amount of time (polynomial). These problems are usually efficiently solvable, although often they are among the hardest problems in $\P$. The canonical example is ``local'' reachability.

\begin{theorem}
\label{thm:polynomial-time-reachability}
Given RS~$\rsA = (S,A)$, two states~$T,U \subseteq S$, and a unary integer~$k$, it is~$\P$-complete to decide if~$T$ leads to~$U$ in at most~$k$ steps.
\end{theorem}

\begin{proof}
We show this result by reduction from the~$\P$-complete bounded halting problem~\cite{Papadimitriou1993a}
\begin{quote}
  Given a deterministic Turing machine~$M$, a string~$x$ and a unary integer~$k$, does~$M$ accept~$x$ within~$k$ steps?
\end{quote}
Without loss of generality, we assume that Turing machine~$M$ has a unique accepting configuration, with the tape head in the accepting state located on the first tape cell, and with an empty tape (i.e., each tape cell is blank). We can build the reaction system~$\rsM$ simulating~$M$ with~$k+1$ cells of tape (the maximum amount of tape exploitable in~$k$ steps) as described in Section~\ref{sec:turing-machines}. Let~$T$ be the encoding of the initial configuration of~$M$ on input~$x$, and let~$U$ be the encoding of the unique accepting configuration. Then, it holds that state~$T$ leads to~$U$ within~$k$ steps (the property to be decided) if and only if~$M$ accepts~$x$ within~$k$ steps.

Since the mapping~$(M,x,k) \mapsto (\rsM,T,U,k)$ can be carried out in logarithmic space, the problem is~$\P$-hard. It also belongs to~$\P$, since 
the property can be checked by simulating~$k$ steps of the reaction system. \qed
\end{proof}

The problem of detecting the presence of ``short'' cycles is also $\P$-complete.

\begin{theorem}
\label{thm:state-in-short-cycle}
Given a RS~$\rsA = (S,A)$, a state~$T \subseteq S$, and a unary integer~$\ell \ge 1$ as input, it is~$\P$-complete to decide if~$T$ is a periodic point of period at most $\ell$.
\end{theorem}

\begin{proof}
We proceed by reduction from the bounded halting problem. Let~$M$ be a Turing machine having a unique accepting configuration in which the accepting state is~$r$, the head is over the tape cell 1, and the tape is empty. Consider the reaction system simulating~$M$ from Section~\ref{sec:turing-machines} and change it so that~$r_1$ is added to the inhibitors of all reactions of types~\eqref{react:transition} and~\eqref{react:preserve-tape}:
\begin{align*}
  &(\set{q_i,a_i}, \set{r_1}, \set{q'_{i+d},a'_i})&
  &\text{for } 1 \le i \le m \\
  &(\set{a_i}, \set{q_i : q \in Q} \cup \set{r_1}, \set{a_i})&
  &\text{for } 1 \le i \le m
\end{align*}
This means that the original reactions are now disabled if~$M$ reaches its accepting state. 
The reaction system~$\rsM$ we are building is also equipped of an additional reaction, which will be instead enabled if~$M$ accepts. Let~$U$ be the state of~$\rsM$ encoding the initial configuration of~$M$ on a given input string~$x$. Such a reaction is:
\[
  (\set{r_1}, \varnothing, U)
\]
Then, the state~$T$, encoding the final configuration of~$M$, belongs to a cycle of length at most~$\ell$ if and only if~$M$ accepts~$x$ within~$\ell-1$ steps.

Since the mapping~$(M,x,k=\ell-1) \mapsto (\rsM,T,\ell)$ can be computed in logarithmic space, the problem under consideration is~$\P$-hard due to the complexity of the bounded halting problem.
The $\P$~upper bound also holds. Indeed, establishing whether any state $T$ is periodic of period at most $\ell$ can be checked 
by performing a simulation of the reaction system with initial state  $T$ and testing whether the state $T$ itself is reached within~$\ell$ steps. \qed
\end{proof}

Even when we search for short cycles with a short preperiod we obtain $\P$-completeness.

\begin{theorem}
\label{thm:reachability-cycle}
Given a RS~$\rsA = (S,A)$, a state~$T \subseteq S$, and two unary integers~$\ell \ge 1$ and~$k$ as input, it is~$\P$-hard to decide if~$T$ has preperiod at most $k$ and period at most~$\ell$. \end{theorem}

\begin{proof}
This problem is~$\P$-hard by trivial reduction, since it is a generalisation of the problem of Theorem~\ref{thm:state-in-short-cycle}, the latter having the fixed value~$k=0$ for the preperiod. A polynomial time algorithm for this problem computes a sequence of~$k+\ell$ steps of~$\rsA$ starting from any state~$T$, and checks that a cycle is reached within the first~$k$ steps. \qed
\end{proof}

\section{Problems in the Polynomial Hierarchy}
\label{sec:hard-problems}

A jump in complexity happens when, instead of checking a local dynamical behaviour that only involves the forward dynamics, we search for the existence of a local behaviour somewhere in the system or the property involves the backward dynamics, i.e., the preimages of a state.

While deciding if a state is part of a short cycle can be done in deterministic polynomial time, establishing the existence of such a state is $\NP$-complete. We will see that this holds even when the length of the cycle is fixed a priori.

In the sequel, for any Boolean formula $\varphi$ in conjunctive or disjunctive normal form and with clauses $\varphi_1, \dots, \varphi_m$, for each clause $\varphi_i$ we denote by $\posvar(\varphi_j)$ (resp., $\negvar(\varphi_i)$) the set of variables appearing in~$\varphi_i$ as positive (resp., negative) literals.
\begin{theorem}
\label{thm:has-short-fixed-cycle}
For each integer~$\ell \ge 1$, it is~$\NP$-complete, on input a RS~$\rsA$, to decide if~$\rsA$ has a periodic point of period at most~$\ell$.
\end{theorem}

\begin{proof}
We prove the~$\NP$-hardness of this problem by reduction from the Boolean satisfiability problem. For any Boolean formula~$\varphi = \varphi_1 \land \cdots \land \varphi_m$ in conjunctive normal form over the set of variables~$X = \set{x_1, \ldots, x_n}$ we build a reaction system~$\rsA = (S,A)$ with set of entities~$S = X \cup Z$, where~$Z = \set{\spadesuit_0,\ldots,\spadesuit_\ell}$. For any state~$T \subseteq S$, the set $T\cap X$ encodes a assignment of~$\varphi$ in which $x_i$ has true value iff~$x_i \in T\cap X$. First of all, the set $A$ contains the following group of reactions
\begin{align}
  \label{react:evaluate-formula}
  &(\negvar(\varphi_j), \posvar(\varphi_j) \cup Z, \set{\spadesuit_0})&
  &\text{for } 1 \le j \le m
\end{align}
Thus, when the system is in a state $T\subseteq X$, this group of reactions produces the entity~$\spadesuit_0$ (\ie at least one among them is enabled in $T$)  if there exists a clause~$\varphi_j$ not satisfied by the assignment encoded by $T$ (hence $\varphi$ itself is not satisfied). Notice that these reactions are disabled in any state $T$ containing some $\spadesuit_t$.

The set $A$ also includes the following reactions   
\begin{align}
  \label{react:preserve-assignment}
  &(\set{x_i}, Z, \set{x_i})&
  &\text{for } 1 \le i \le n
\end{align}
which preserve the value of the assignment of $\varphi$ encoded by a current state~$T$ in the next state when $T \subseteq X$.

In this way, the reactions of type~\eqref{react:evaluate-formula} and~\eqref{react:preserve-assignment} ensure that~$\res_\rsA(T) = T$ for any~$T \subseteq X$ encoding an assignment satisfying~$\varphi$, and~$\res_\rsA(T) = T \cup \set{\spadesuit_0}$ if~$T \subseteq X$ but~$T$ codifies an assignment which does not satisfy~$\varphi$.

The following last group of reactions constituting $A$ creates a cycle of length~$\ell+1$ formed exactly by the entities of~$Z$
\begin{align}
  \label{react:spadesuit-cycle}
  &(\set{\spadesuit_t}, \set{\spadesuit_s : 0 \le s < t}, \set{\spadesuit_{(t+1) \bmod (\ell+1)}})&
  &\text{for } 0 \le t \le \ell
\end{align}
Indeed, it trivially holds that~$\res_\rsA(\set{\spadesuit_t}) = \set{\spadesuit_{(t+1) \bmod (\ell+1)}}$. Moreover, if the system is in a state $T$ with $T \cap Z \ne \varnothing$, then its next state is~$\set{\spadesuit_{t+1 \bmod \ell+1}}$, where~$t = \min \set{s : \spadesuit_s \in T}$. Hence, the reaction system~$\rsA$ has both a cycle of length~$\ell+1$ (namely~$\set{\spadesuit_0} \to \cdots \to \set{\spadesuit_\ell} \to \set{\spadesuit_0}$) and at least a fixed point, that is a cycle of length at most~$\ell$, if and only if the formula~$\varphi$ admits a satisfying assignment.

Since the mapping~$\varphi \mapsto \rsA$ can be computed in polynomial time, the problem is~$\NP$-hard. The membership in~$\NP$ follows by the existence of a non-deterministic algorithm guessing an initial state~$T$ and checking in polynomial time whether~$T$ is again reached within~$\ell$ steps. \qed
\end{proof}

The difficulty of the problem does not increase when the length of the cycle is given in input.

\begin{theorem}
\label{thm:has-short-cycle}
Given a RS~$\rsA$ and a unary integer~$\ell \ge 1$ as input, it is~$\NP$-complete to decide if~$\rsA$ has a periodic point of period at most~$\ell$.
\end{theorem}

\begin{proof}
The problem is~$\NP$-hard by reduction from any of the problems of Theorem~\ref{thm:has-short-fixed-cycle}, which are~$\NP$-hard for every fixed~$\ell \ge 1$. The same algorithm mentioned at the end of the proof of Theorem~\ref{thm:has-short-fixed-cycle} proves the membership in~$\NP$, since the extra input~$\ell$ is given in unary notation. \qed
\end{proof}

Checking if a state is part of a short cycle can be performed in deterministic polynomial time. However, checking if there exists another state reaching it after a certain amount of steps makes the problem $\NP$-complete.

\begin{theorem}
\label{thm:T-in-short-cycle-with-long-handle}
Given a RS~$\rsA = (S,A)$, a state~$T \subseteq S$, and two unary integers~$\ell \ge 1$ and~$k$ as input, it is~$\NP$-complete to decide if~$T$ is a periodic point of period at most~$\ell$ and there exists a state~$U \subseteq S$ leading to~$T$ with preperiod at least~$k$.
\end{theorem}

\begin{proof}
We proceed by reduction from the Boolean satisfiability problem. 
For any Boolean formula~$\varphi = \varphi_1 \land \cdots \land \varphi_m$ in conjunctive normal form over the set of variables~$X = \set{x_1, \ldots, x_n}$ we build a reaction system~$\rsA = (S,A)$ with set of entities~$S = X \cup C$, where~$C = \set{\varphi_1, \ldots, \varphi_m}$. As in the proof of Theorem~\ref{thm:has-short-fixed-cycle}, for any state~$T \subseteq S$, the set $T\cap X$ encodes a truth assignment of~$\varphi$. The set $A$ contains the following two groups of reactions
\begin{align*}
  &\big( \set{x_i}, \varnothing, \set{\varphi_j : x_i \in \posvar(\varphi_j)} \big)&
  &\text{for } 1 \le i \le n \\
  &\big( \varnothing, \set{x_i}, \set{\varphi_j : x_i \in \negvar(\varphi_j)} \big)&
  &\text{for } 1 \le i \le n
\end{align*}
By evaluating $\varphi$, these reactions map a state~$T$ to the set of clauses satisfied by the assignment encoded by~$T \cap X$, ignoring any element of~$T \cap C$.

The further reaction in $A$
\begin{align*}
  (C, X, C)
\end{align*}
ensures that the state consisting exactly of all clauses is a fixed point, since these latter are preserved when  they appear all together without any variable from~$X$.

Then, the resulting reaction system~$\rsA $ has a fixed point~$T=C$, that is a cycle of length at most~$\ell=1$, with a state~$U$ leading to~$T$ and having preperiod at least~$k=1$, if and only if~$U \cap X$ encodes a assignment which satisfies~$\varphi$. 

Since the mapping~$\varphi \mapsto \rsA$ is polynomial-time computable, the problem is thus~$\NP$-hard. A polynomial-time nondeterministic algorithm for the problem under consideration exists. It guesses a state~$U$ leading to~$T$ in exactly~$k$ steps, verifies if~$T$ belongs to a cycle of length at most~$\ell$, and checks that the set~$\set{U,  \res_{\rsA}(U), \ldots, T}$ and the cycle only have~$T$ as a common state. Therefore, the membership in $\NP$ follows.
 \qed
\end{proof}

More in general, the problem of exploring the dynamics backwards is $\NP$-complete.

\begin{theorem}
\label{thm:k-ancestor}
Given a RS~$\rsA = (S,A)$, a state~$T \subseteq S$, and a unary integer~$k$ as input, it is~$\NP$-complete to decide if~$T$ has a~$k$-ancestor.
\end{theorem}

\begin{proof}
Perform the same reduction as in the proof of Theorem~\ref{thm:T-in-short-cycle-with-long-handle}, except adding the entity~$\heartsuit$ to~$S$ and changing the reaction~$(C,X,C)$ to~$(C, X, C \cup \set{\heartsuit})$. Now the obtained reaction system has the fixed point $C \cup \set{\heartsuit}$ rather than~$C$, and~$C$ leads to~$C \cup \set{\heartsuit}$ in one step. Then, state~$C$ has~$T$ as preimage, \ie $T$ is a~$1$-ancestor of $C$, if and only if~$T \cap X$ satisfies~$\varphi$ for some set~$T \subseteq S$. Thus, the $\NP$-hardness follows. The problem is also in~$\NP$. Indeed, once a state~$U \subseteq S$ is guessed, it can be checked in polynomial time whether~$\res_\rsA^k(U) = T$. \qed
\end{proof}

Establishing if there exists a state that leads to a short cycle after a given number of steps, i.e., a state far enough from the cycle, turns out be another $\NP$-complete problem.

\begin{theorem}
\label{thm:reachability-cycle-exists}
Given a RS~$\rsA = (S,A)$ and two unary integers~$\ell \ge 1$ and~$k$ as input, it is~$\NP$-complete to decide if there exists a state~$T \subseteq S$ with period at most~$\ell$ and preperiod at least~$k$.
\end{theorem}

\begin{proof}
The problem from Theorem~\ref{thm:has-short-cycle} is a special case of the one under consideration and is obtained when~$k=0$. Hence, by reduction from the former problem,  the one we are dealing with is~$\NP$-hard too. Its membership in~$\NP$ follows by the existence of an algorithm similar to that at the end of the proof of Theorem~\ref{thm:T-in-short-cycle-with-long-handle}, except that the state~$T$ must also be guessed, rather than being given as input. \qed
\end{proof}

Checking if every state of a reaction system has exactly one preimage is the same as deciding if the result function is a bijection. This problem can be shown to be $\coNP$-complete by a variation of a proof by Ehrenfeucht and Rozenberg~\cite{Ehrenfeucht2007a}.

\begin{theorem}
\label{thm:bijection}Given a RS~$\rsA = (S,A)$, it is~$\coNP$-complete to decide if~$\res_\rsA$ is a bijection, that is, if every state of~$\rsA$ is periodic.
\end{theorem}

\begin{proof}
We prove~$\coNP$-hardness by reduction from the tautology problem for Boolean formulae in disjunctive normal form~\cite{Papadimitriou1993a}. Given a formula~$\varphi = \varphi_1 \lor \ldots \lor \varphi_m$ over the variables~$X = \set{x_1, \ldots, x_n}$, we build the reaction system~$\rsA = (S,A)$ where~$S = X \cup \set{\heartsuit}$ and the reactions are 
\begin{align}
  \label{react:evaluate-dnf}
  &(\posvar(\varphi_j) \cup \set{\heartsuit}, \negvar(\varphi_j), \set{\heartsuit})&
  &\text{for } 1 \le j \le m \\
  \label{react:preserve-x-dnf}
  &(\set{x_i}, \varnothing, \set{x_i})&
  &\text{for } 1 \le i \le n
\end{align}
As usual, for any state $T$, $T\cap X$ encode a truth assignment of $\varphi$. In this way, if~$\heartsuit$ occurs in the currente state, the reactions of type~\eqref{react:evaluate-dnf} evaluate each conjunctive clause (which is satisfied iff all positive variables are set to true and all negative ones to false) and preserve the element~$\heartsuit$  if the clause (and thus the whole formula) is satisfied. The reactions of type~\eqref{react:preserve-x-dnf} preserve all variables in the current state.

Hence, the behaviour of~$\rsA$ is as follows. If the current state~$T \subseteq X$, only reactions of type~\eqref{react:preserve-x-dnf} are enabled, and thus~$\res_\rsA(T) = T$. On the other hand, if~$\heartsuit\in T$, we have~$\res_\rsA(T \cup \set{\heartsuit}) = T \cup \set{\heartsuit}$ if at least one reaction of type~\eqref{react:evaluate-dnf} is enabled, that is, if~$\varphi$ is satisfied by the assignment encoded by~$T$; otherwise, it holds that~$\res_\rsA(T \cup \set{\heartsuit}) = T$. As a consequence, each state of~$\rsA$ is a fixed point, and thus a cycle, iff~$\varphi$ is a tautology; otherwise, there exists a state with two preimages, namely, a state~$T \subseteq X$ encoding a non-satisfying assignment and having both~$T$ and~$T \cup \set{\heartsuit}$ as preimages.

The considered problem belongs to~$\coNP$, since there exists a non-deterministic algorithm guessing two distinct states~$T,U \subseteq S$ and checking whether~$\res_\rsA(T) = \res_\rsA(U)$ in polynomial time. \qed
\end{proof}

The previous construction actually proves another similar result:

\begin{theorem}
\label{thm:identity}
Given a RS~$\rsA = (S,A)$, it is~$\coNP$-complete to decide if~$\res_\rsA$ is the identity function, that is, if each state is a fixed point. \qed
\end{theorem}

Establishing whether a local attractor cycle has a small attraction basin is $\coNP$-complete. 

\begin{theorem}
\label{thm:small-basin-for-T}
Given a RS~$\rsA = (S,A)$, a state~$T \subseteq S$, and two unary integers~$\ell \ge 1$ and~$d$ as input, it is~$\coNP$-complete to decide if~$T$ belongs to a local attractor cycle of length at most~$\ell$ with attraction basin of diameter at most~$d$.
\end{theorem}

\begin{proof}
We use a reduction from tautology similar to that introduced in the proof of Theorem~\ref{thm:bijection}. For any formula ~$\varphi = \varphi_1 \lor \ldots \lor \varphi_m$ over the variables~$X = \set{x_1, \ldots, x_n}$, consider the reaction system from the aforementioned reduction.  We add an extra entity~$\spadesuit$ to its background set and slightly change the original reactions as follows:
\begin{align}
  \label{react:evaluate-dnf-2}
  &(\posvar(\varphi_j) \cup \set{\heartsuit}, \negvar(\varphi_j) \cup \set{\spadesuit}, \set{\heartsuit})&
  &\text{for } 1 \le j \le m \\
  \label{react:preserve-x-dnf-2}
  &(\set{x_i,\heartsuit}, \set{\spadesuit}, \set{x_i})&
  &\text{for } 1 \le i \le n
\end{align}
Furthermore, we add the two extra reactions 
\begin{align*}
(\varnothing, \set{\heartsuit}, \set{\spadesuit})\\
(\set{\spadesuit}, \varnothing, \set{\spadesuit})
\end{align*}
producing~$\spadesuit$ whenever~$\heartsuit$ is missing, and preserving~$\spadesuit$, respectively.

The resulting reaction system~$\rsA$ behaves as follows. Let~$U \subseteq X$. By evaluating $\varphi$ on the basis the assignment encoded by $U$, it follows that~$\res_\rsA(U \cup \set{\heartsuit}) = U \cup \set{\heartsuit}$ if~$U$ encodes an assignment satisfying~$\varphi$, and~$\res_\rsA(U \cup \set{\heartsuit}) = U$ otherwise. Furthermore, either when~$\spadesuit$ belongs to the current state (which is thus either $U \cup \set{\spadesuit}$ or $U \cup \set{\spadesuit,\heartsuit}$), or when neither~$\spadesuit$ nor~$\heartsuit$ do (\ie the current state is $U$), it holds that~$\res_\rsA(U) = \res_\rsA(U \cup \set{\spadesuit}) = \res_\rsA(U \cup \set{\spadesuit,\heartsuit}) = \set{\spadesuit}$. Hence, state~$\set{\spadesuit}$ has a~$2$-ancestor if and only if the formula has a \emph{non}-satisfying assignment~$U$. In that case, we have~$\res_\rsA^2(U \cup \set{\heartsuit}) = \res_\rsA(U \cup \set{\spadesuit}) = \set{\spadesuit}$. By letting~$T = \set{\spadesuit}$, $\ell=1$, and~$d=1$, and since the mapping~$\varphi \mapsto (\rsA,T,\ell,d)$ is polynomial-time computable, we obtain the~$\coNP$-hardness of the problem.

The problem is in~$\coNP$ since once guessed a state~$U$ it is possible to check in polynomial time if it falsifies the required property.
\qed
\end{proof}

On the other hand, deciding if a small attractor cycle has a large attraction basin is an $\NP$-complete problem.

\begin{theorem}
\label{thm:large-basin-for-T}
Given a RS~$\rsA = (S,A)$, a state~$T \subseteq S$, and two unary integers~$\ell \ge 1$ and~$d$ as input, it is~$\NP$-complete to decide if~$T$ belongs to a local attractor cycle of length at most~$\ell$ with attraction basin of diameter at least~$d$.
\end{theorem}

\begin{proof}
The state~$\set{\spadesuit}$ in the proof of Theorem~\ref{thm:small-basin-for-T} has an attraction basin of diameter at least~$2$ iff the Boolean formula~$\varphi$ is \emph{not} a tautology, and the latter is the statement of a~$\NP$-complete problem (since its negation is a~$\coNP$-complete one). Indeed, the same algorithm as in the proof of Theorem~\ref{thm:small-basin-for-T}, but with reversed acceptance and rejection, proves the membership in~$\NP$. \qed
\end{proof}

Establishing the existence of a small attractor cycle with a large attraction basin remains $\NP$-complete.

\begin{theorem}
\label{thm:large-basin-exists}
Given a RS~$\rsA = (S,A)$ and two unary integers~$\ell \ge 1$ and~$d$ as input, it is~$\NP$-complete to decide if there exists a local attractor cycle of length at most~$\ell$ with attraction basin of diameter at least~$d$.
\end{theorem}

\begin{proof}
Consider the reactions system~$\rsA$ from the proof of Theorem~\ref{thm:large-basin-for-T}; all of its states are either fixed points or eventually reach state~$\set{\spadesuit}$. Hence,~$\set{\spadesuit}$ is the only state with attraction basin of diameter at least~$1$. Hence, asking whether there exists a fixed point (i.e., a cycle of length~$\ell=1$) with an attraction basin of diameter at least~$d=2$ is the same as asking whether~$T = \set{\spadesuit}$ has this property, which is~$\NP$-complete. The problem under consideration is in~$\NP$, since a state~$T$ can be guessed before the same algorithm in the proof of Theorem~\ref{thm:large-basin-for-T} is used.
\end{proof}

Differently from the previous case, the problem of establishing the existence of a small attractor cycle with a small attraction basin is higher in the polynomial hierarchy.

\begin{theorem}
\label{thm:small-basin-exists}
Given a RS~$\rsA = (S,A)$ and two unary integers~$\ell \ge 1$ and~$d$ as input, it is~$\SigmaTwoP$-complete to decide if there exists a local attractor cycle of length at most~$\ell$ with attraction basin of diameter at most~$d$.
\end{theorem}

\begin{proof}
First of all, we show the membership in~$\SigmaTwoP$ of the problem.  There exists an alternating Turing machine able to guess a state~$T \subseteq S$ and then, by iterating the~$\res_\rsA$ function, check that it belongs to a cycle of length at most~$\ell$. Subsequently, the machine also guesses a state~$U \subseteq S$ and check if it reaches that cycle in~$d+1$ steps (this would prove that the attraction basin has diameter larger than~$d$; if no such guess is possible, then the attraction basin has diameter at most $d$). If this is the case, the machine rejects, while it accepts, otherwise; the overall result of the machine is acceptance if and only if the answer to the considered problem is true. Therefore, the membership in~$\SigmaTwoP$ follows.

We now prove the~$\SigmaTwoP$-hardness of the problem by reduction from~$\exists\forall\SAT$, i.e., the problem of deciding if~$\exists X \forall Y \varphi(X,Y)$, where~$\varphi$ is a Boolean formula in disjunctive normal form with conjunctive clauses~$\varphi_1, \ldots, \varphi_k$ over the variables~$X = \set{x_1, \ldots, x_m}$ and~$Y = \set{y_1, \ldots, y_n}$.

For any of such formulae $\varphi$,  define the reaction system~$\rsA = (S,A)$ where~$S = X \cup Y \cup \set{\heartsuit,\spadesuit}$ and~$A$ consists of the two group of reactions
\begin{align}
  &(\posvar(\varphi_j), \negvar(\varphi_j), \set{\heartsuit})&
  &\text{for each clause } \varphi_j
  \label{react:evaluate}\\
  &(\set{x_i}, \varnothing, \set{x_i})&
  &\text{for each variable } x_i \in X
  \label{react:preserve-x}
\end{align}
which depend on the Boolean formula, and the following further reactions
\begin{align}
  &(\varnothing, \varnothing, \set{\spadesuit})
  \label{react:make-spade} \\
  &(\set{\heartsuit}, \varnothing, \set{\heartsuit})
  \label{react:preserve-heart} \\
  &(\set{\spadesuit}, \varnothing, \set{\heartsuit})
  \label{react:spade-makes-heart}
\end{align}
We are going to describe the behaviour of $\rsA$. First of all, consider a state of the form~$X_1 \cup Y_1$ for some~$X_1 \subseteq X$ and~$Y_1 \subseteq Y$. As usual, such a state represents an assignment of $\varphi$, where the variables appearing in~$X_1 \cup Y_1$ have true value, and the missing ones have false value. A reaction of type~\eqref{react:evaluate} is then enabled if and only if the corresponding conjunctive clause~$\varphi_j$ is satisfied. Furthermore, a reaction of type~\eqref{react:preserve-x} is enabled for each element of~$X_1$, while the reaction~\eqref{react:make-spade} is always enabled. As a consequence, we have~$\res_\rsA(X_1 \cup Y_1) = X_1 \cup \set{\heartsuit,\spadesuit}$ if at least one clause (and thus the disjunction~$\varphi$) is satisfied by the assignment represented by~$X_1 \cup Y_1$; otherwise, we have~$\res_\rsA(X_1 \cup Y_1) = X_1 \cup \set{\spadesuit}$.

Furthermore, we have
\begin{align*}
  \res_\rsA(X_1 \cup Y_1 \cup \set{\heartsuit}) =
  \res_\rsA(X_1 \cup Y_1 \cup \set{\spadesuit}) =
  \res_\rsA(X_1 \cup Y_1 \cup \set{\spadesuit,\heartsuit}) =
  X_1 \cup \set{\spadesuit,\heartsuit}
\end{align*}
irrespective of whether the assignment encoded by~$X_1 \cup Y_1$ satisfies~$\varphi$.

This means that all~$2^{|X|}$ states of the form~$X_1 \cup \set{\spadesuit,\heartsuit}$ for any~$X_1 \subseteq X$ are fixed points and~$\rsA$ admits no further fixed points or cycles. Moreover, each fixed point~$X_1 \cup \set{\spadesuit,\heartsuit}$ attracts in one step all states of the forms~$X_1 \cup Y_1 \cup \set{\heartsuit}$, $X_1 \cup Y_1 \cup \set{\spadesuit}$, and~$X_1 \cup Y_1 \cup \set{\spadesuit,\heartsuit}$ for any~$Y_1 \subseteq Y$. On the other hand, the states of the form~$X_1 \cup Y_1$ reach~$X_1 \cup \set{\spadesuit,\heartsuit}$ in one step iff~$X_1 \cup Y_1$ satisfies~$\varphi$, and in two steps, otherwise.

In conclusion, there exists a fixed point (which is a cycle of length at most~$\ell=1$), necessarily of the form~$X_1 \cup \set{\spadesuit,\heartsuit}$ for some~$X_1 \subseteq X$, with an attraction basin of diameter at most~$d=1$ if and only if~$X_1 \cup Y_1$ satisfies~$\varphi$ for all~$Y_1 \subseteq Y$. This is equivalent to~$\forall X \exists Y \varphi(X,Y)$ being true. Since the mapping~$\varphi \mapsto (\rsA, \ell, d)$ can be computed in polynomial time, the problem is~$\SigmaTwoP$-hard. \qed
\end{proof}

\section{Problems Solvable in Polynomial Space}
\label{sec:very-hard-problems}

Removing the restrictions on the length of cycles or preperiods leads to 
properties involving a potentially  exponential number of states, and algorithms checking these properties might need to explore all of them. Such algorithms, however, can be designed to use
a polynomial amount of space, since every state contains only a polynomial amount of entities.

The canonical $\PSPACE$-complete problem for reactions systems is reachability.

\begin{theorem}
\label{thm:reachability}
Given RS~$\rsA = (S,A)$ and two states~$T,U \subseteq S$, it is~$\PSPACE$-complete to decide if~$T$ leads to~$U$.
\end{theorem}

\begin{proof}
We prove the hardness of the problem by reduction from the following~$\PSPACE$-complete problem~\cite{Papadimitriou1993a}:
\begin{quote}
Given a deterministic Turing machine~$M$, a string~$x$ and a unary integer~$m$, does~$M$ accept~$x$ without using more than~$m$ tape cells?
\end{quote}
As in the proof of Theorem~\ref{thm:polynomial-time-reachability}, for any~$(M,x,m)$  we build a reaction system~$\rsA$ by exploiting the construction described in Section~\ref{sec:turing-machines} with $k=m$. By encoding again the initial configuration of~$M$ on input~$x$ as a state~$T$ of~$\rsA$ and the unique accepting configuration of~$M$ as state~$U$, we obtain the desired reduction. Indeed, if~$M$ accepts~$x$ without exceeding~$m$ tape cells, then $T$ leads to $U$. Furthermore, this never happens
if~$M$ rejects, or fails to halt, or moves the tape head on cell~$m+1$; remark that in the latter case, $T$ does not lead to $U$ since the reactions of type~\eqref{react:transition} are not defined for~$i=m+1$, and the element~$q_{m+1}$ disappears.

The problem is in~$\PSPACE$ because there exists an algorithm able to check in polynomial space whether~$T$ leads to~$U$ (such an algorithm rejects if this does not happen within~$2^{|S|}$ steps, that is the number of configurations of~$\rsA$). \qed
\end{proof}

To prove the $\PSPACE$-hardness of other problems, in this section we are going to introduce a variant of the polynomial space simulation of the Turing machine from Section~\ref{sec:turing-machines} used in the proof of Theorem~\ref{thm:reachability}. 
Since most of the statements in the next theorems involve a quantification across \emph{every} state of the reaction system, the dynamics of all the possible states will have to be considered and governed. To this extent, for any Turing machine $M$ with state set $Q$ where $f\in Q$ is the final accepting state, tape alphabet $\Sigma$, transition function $\delta$, any input $x$, and any space bound $m$ given in unary, the new construction of the reaction system and the encoding of the configurations of $M$ into related states is based on what follows:
\begin{enumerate}
  \item \label{tm:initial-point} Let $T$ be the state of the reaction system encoding the initial configuration of $M$. As in the proof of Theorem~\ref{thm:reachability}, $M$ accepts by ending up in a specific accepting configuration with the tape head on the first tape and encoded by the state $U$ in the reaction system which  will be now a fixed point; otherwise, $M$ rejects by diverging.
  \item  \label{tm:second-point}
  All states of the reaction system encoding a configuration of $M$ either lead to the fixed point $U$ (which corresponds to the accepting configuration reached by $M)$ or enter a cycle. We stress that this set of states also includes the ones encoding configurations of $M$ that are not reachable from the initial configuration. A timer is now coupled to the encoding of configurations of $M$, in order to force all states of the reaction systems belonging to this set to reach $T$ once the timer reaches an appropriate value.
  \item \label{tm:final-point} If a state of the reaction system does not encode a valid Turing machine configuration and a valid timer, it is detected and forced to reach $T$. Then, its dynamics will evolve as described in statement~\ref{tm:initial-point}.
\end{enumerate}
In this way, \emph{all states} lead to the fixed point $U$ only if $T$ itself leads to $U$ (\ie if $M$ reaches the accepting configuration within the space bound $m$).
This is accomplished by building the new reaction system $\rsM=(S, A)$ as follows.

Consider the background set of the reaction system from Section~\ref{sec:turing-machines} and keep unchanged the way to encode any configuration of $M$ into a state of reaction system state.  Now, an entity is added to that background set, namely the entity $\restart$ which represents the ``reset'' to be executed in some cases from a state of $\rsM$ to the state $T$, which has well defined dynamics and encoding the initial configuration of $M$.

Reminding that $f_1$ is the entity corresponding to the accepting state $f$ in the accepting configuration of $M$ (with the tape head located on the first cell), we introduce the following reaction in order to force the element $f_1$ to constitute the fixed point $U=\{f_1\}$ when it is the only entity of the state\footnote{Actually, $U$ encodes the state of $\rsM$ coming immediately after the one encoding the accepting configuration of $M$. However, the two state states coincide if one ignores the content of the tape and the direction.}:
\begin{align*}
  &(\set{f_1}, \varnothing, \set{f_1})
\end{align*}

Since we are interested only in accepting computations within a given space bound $m$, the set of reactions $A$ also contains the following reactions, which force a ``reset'' by generating the entity $\restart$ if the tape head reaches the $(m+1)$-th cell:
\begin{align*}
  &(\set{q_{m+1}}, \set{f_1,\restart}, \set{\restart})&
  &\text{for } q \in Q
\end{align*}

In order to deal with the simulation of $M$ by $\rsM$ upon reaching an accepting configuration or the situation when a reset is needed, reactions~\ref{react:transition} and~\ref{react:preserve-tape} from Section~\ref{sec:turing-machines} implementing every transition $\delta(q,a)=(q', a', d)$ appear now in $A$ with $f_1$ and $\restart$ as additional inhibitors:
\begin{align*}
  &(\set{q_i,a_i}, \set{f_1,\restart}, \set{q'_{i+d},a'_i})&
  &\text{for } 1 \le i \le m \\
  &(\set{a_i}, \set{q_i : q \in Q} \cup \set{f_1,\restart}, \set{a_i})&
  &\text{for } 1 \le i \le m
\end{align*}
In this way, the transition of $M$ is not executed and the tape is not preserved, thus deleting the current configuration of the simulated Turing machine.

We now describe how $\restart$ is produced by a timer in order to comply with property~\ref{tm:second-point} on page~\pageref{tm:second-point}. Denote by $p_k$ the $k$-th prime number and define $K$ as the smallest number such that $\prod_{k=1}^K p_k -1$ is greater than the number of possible distinct configuration of $M$ when working in space $m$. For  each $k=1, \ldots, K$, let $c_{k,0},\ldots,c_{k,p_k-1}$ be further entities in the background set $S$ of $\rsM$. For each $k$, the following group of reactions is included in $A$:
\begin{align*}
  &(\set{c_{k,t}}, \set{f_1,\restart}, \set{c_{k,(t+1) \bmod p_k}})&
  &\text{for } 0 \le t < p_k
\end{align*}
Each group of reactions defines a cycle of length $p_k$. When all the reactions of all groups act together every time step of the evolution of the reaction system in which $f_1$ and $\restart$ do not appear, they give rise to a cycle of length $L=\prod_{k=1}^K p_k$, since the $p_k$'s are pairwise coprime. In this way, the dynamical evolution of $\rsM$ starting from any state including $Y_0 = \set{c_{1,0},\ldots,c_{K,0}}$ and containing neither $f_1$ nor $\restart$ is associated with a timer whose value at time $t$ can be identified by $Y_t=\set{c_{1,t_1},\ldots,c_{K,t_K}}$ where $t_k=t \bmod p_k$. The timer reaches its maximum value once the state of $\rsM$ contains $Y_t=\set{c_{1,p_1-1},\ldots,c_{K,p_K-1}}$, and this happens just at time $t=L-1$.

In the following two steps, the entity $\restart$ is produced and the timer is reset to $0$, the state of the reactions system $\rsM$ is set to $C_0=T\cup Y_0$ as desired, and the simulation of $M$ by $\rsM$ restarts from the initial configuration encoded by $T$. This is accomplished by the following two reactions:
\begin{align*}
  &(\set{c_{k,p_k-1} : 1 \le k \le K}, \set{f_1, \restart}, \set{\restart})\\
  &(\set{\restart}, \set{f_1}, C_0)
\end{align*}
Remark that, if $M$ works in space $m$, then it has $|\Sigma|^m \times m \times |Q|$ possible distinct configurations. By letting $K = \lceil \log_2 \left( |\Sigma|^m \times m \times |Q| \right) \rceil + 1$, it holds that $\prod_{k=1}^K p_k - 1 = L - 1$ is larger that the number of possible configurations of $M$, since $p_k \ge 2$ for all $k$. By the Prime Number Theorem~\cite{Dudley1969a}, the $K$-th prime is asymptotically $K \ln K$. To find $K$ primes, it is therefore sufficient to test the primality of the first $O(K \ln K)$ natural numbers, where each number can be tested even using a brute force algorithm. Since this value is polynomial with respect to length of the description of $M$ and with respect to $m$ (which is given in unary), the construction of $\rsM$ we are describing can be carried out in polynomial time.

In order to comply with statement~\ref{tm:final-point} on page~\pageref{tm:final-point}, the set of reactions $A$ must also contain reactions generating $\restart$ each time the current state of the reaction systems fails to encode either a valid pair $($\emph{timer, configuration of $M$}$)$ nor it is the fixed point $\set{f_1}$. First of all, the following reactions generate $\restart$ when the current state of $\rsM$ contains zero or more than one entity encoding the pair $($\emph{state, head position of $M$}$)$:
\begin{align*}
  &(\varnothing, \set{q_i : q \in Q, 1 \le i \le m} \cup \set{f_1,\restart}, \set{\restart})\\
  &(\set{q_i,r_j}, \set{f_1,\restart}, \set{\restart})&
  &\text{for } q,r \in Q \text{ and } 1 \le i,j \le m
\end{align*}
The entity $\restart$ must also be produced each time the state of $\rsM$ does not encode exactly one symbol per tape cell:
\begin{align*}
  &(\varnothing, \set{a_i : a \in \Sigma} \cup \set{f_1,\restart}, \set{\restart})&
  &\text{for } 1 \le i \le m\\
  &(\set{a_i,b_i}, \set{f_1,\restart}, \set{\restart})&
  &\text{for } a,b \in \Sigma \text{ and } 1 \le i \le m
\end{align*}
Finally, there must also be reactions in $A$ producing $\restart$ when the state does not include a well-formed timer. That is, when for a single prime $p_k$ there does not appear exactly one entity of the form $c_{k,i}$ in the current state of $\rsM$:
\begin{align*}
  &(\varnothing, \set{c_{k,t} : 0 \le t < p_k} \cup \set{f_1,\restart}, \set{\restart})&
  &\text{for } 1 \le k \le K \\
  &(\set{c_{k,t},c_{k,s}}, \set{f_1,\restart}, \set{\restart})&
  &\text{for } 1 \le k \le K \text{ and } 0 \le t,s < p_k
\end{align*}

In this way, we have obtained a construction complying with properties~\ref{tm:initial-point}--\ref{tm:final-point} on page~\pageref{tm:initial-point} and that can be carried out in polynomial time.

We are now able to prove that checking the existence of long cycles is $\PSPACE$-complete.

\begin{theorem}
\label{thm:has-large-fixed-cycle}
Given a RS~$\rsA$ over a background set $S$ and a unary integer~$\ell \ge 1$ as input, it is~$\PSPACE$-complete to decide if~$\rsA$ has a periodic point of period at least~$\ell$.
\end{theorem}
\begin{proof}
  The problem is in $\PSPACE$ since there exists an algorithm that iterates over all the possible states of~$\rsA$ and checks if each of them leads back to itself after at least $\ell$ steps (but which is necessarily not greater than $2^{|S|})$. This can be performed in polynomial space.

  We prove that the problem under consideration is $\PSPACE$-hard by a reduction from the same problem used in the proof of Theorem~\ref{thm:reachability}. For any Turing machine $M$, any input string $x$, and any space bound $m$, consider the reaction system $\rsM$ from the above construction.
  Let $\ell > 1$ be the number of possible configurations for $M$ when working in space $m$.
  By construction, it holds that if $M$ accepts on input $x$, then each state of $\rsM$ 
  leads to the fixed point $\set{f_1}$, and, in particular, no cycle of length at least $\ell$ can be reached. Otherwise, if $M$ rejects on input $x$ by diverging, there is a cycle containing the state $C_0$, which includes the encoding of the initial configuration of $M$. The length of this cycle is dictated by the time needed for the timer to generate the entity $\restart$, which, by hypothesis, is greater than the number of configuration of $M$ when working in space $m$, and, therefore, of $\ell$. Hence, there is cycle of length at least $\ell$ if and only if $M$ does not accept on input $x$ working in space $m$. This concludes the reduction and the $\PSPACE$-hardness of the problem under consideration follows. \qed
\end{proof}

Also checking whether there exists a state that is far enough from a long cycle is $\PSPACE$-complete.

\begin{theorem}
\label{thm:reachability-large-cycle-exists}
Given a RS~$\rsA = (S,A)$ and two unary integers~$\ell \ge 1$ and~$k$ as input, it is~$\PSPACE$-complete to decide if there exists a state~$T \subseteq S$  with period at least~$\ell$ and preperiod at least~$k$.
\end{theorem}
\begin{proof}
  The problem is in $\PSPACE$ since there exists a non-deterministic algorithm  able to guess two states $T$ and $U$ and check in polynomial space that $T$ leads to $U$ in at least $k$ steps, that $U$ is in a cycle of length at least $\ell$, and that $U$ is the first state of that cycle visited by $\rsA$ starting from $T$. Since $\PSPACE$ is closed under non-determinism, it contains the problem under consideration. 
 
  Such a problem is $\PSPACE$-hard since in the case $k=0$ it is equivalent to the one from Theorem~\ref{thm:has-large-fixed-cycle}, which is $\PSPACE$-complete. \qed
\end{proof}

The same problem remains $\PSPACE$-complete even when the starting state is given in input.

\begin{theorem}
  \label{thm:cycle-with-long-handle}
  Given a RS~$\rsA = (S,A)$, a state $T \subseteq S$ and two unary integers $\ell \ge 1$ and $k$ as input, it is $\PSPACE$-complete to decide if $T$ has preperiod at least $k$ and period at least $\ell$ (resp., at most $\ell$).
\end{theorem}
\begin{proof}
  The problem is in $\PSPACE$ since it is possible to check in polynomial space if each state in the dynamics of $\rsA$ starting from $T$ is part of a cycle or not, thus obtaining the preperiod and period of $T$.

  The $\PSPACE$-hardness follows by a reduction from the same problem used in the proof of Theorem~\ref{thm:reachability}. For any Turing machine $M$, input string $x$, and unary space bound $m$, let $\rsM$ be the reaction system obtained by the construction presented in this section.  Take $T = C_0$, $k=1$, and $\ell=1$. It holds that $T$ has preperiod at least one if and only if $M$ accepts on input $x$. In that case, the dynamics reaches, after at least one step, the fixed point $\set{f_1}$, and thus $T$ has period $1$, otherwise $T$ is part of a cycle and thus has preperiod $0$. Then, the reduction is accomplished.\qed
\end{proof}

All problems related to the existence of global attractors also turn out to be $\PSPACE$-complete.

\begin{theorem}
\label{thm:small-global}
Given a RS~$\rsA = (S,A)$ and a unary integer~$\ell \ge 1$ as input, it is~$\PSPACE$-complete to decide:
\begin{enumerate}
  \item If a state~$T \subseteq S$, also given as input, belongs to a global attractor cycle of length at most~$\ell$ (resp., at least~$\ell$).
  \item If there exists a global attractor cycle of length at most~$\ell$ (resp., at least~$\ell$).
\end{enumerate}
\end{theorem}
\begin{proof}
  All these problems are in $\PSPACE$, since it is possible to check in polynomial space if any given state $T$ of $\rsA$ belongs to a cycle of length at most~$\ell$ (resp., at least~$\ell$) and it is reachable from every other state. The existential version of the problems can be solved by iterating across all the states of $\rsA$, which can also be performed in polynomial space.

  The $\PSPACE$-hardness again follows from the reduction from the same problem used in the proof of Theorem~\ref{thm:reachability}. For any Turing Machine $M$, input string $x$, and unary space bound $m$, according to the construction presented in this section, the reaction system $\rsM$ has a global attractor cycle, namely the fixed point $\set{f_1}$, if and only if $M$ accept on input $x$ within space $m$. The reduction is completed by setting $\ell = 1$ and, for the first problem, $T=\set{f_1}$.
  \qed
\end{proof}

\section{Conclusions}\label{sec:conclusion}

This paper has tackled the study of the computational complexity of problems associated with the dynamics of RS. 
Table~\ref{tab:comprehensive} provides a reference list of such problems.

We stress that the interest of such studies is twofold. From one hand, results about RS are lower bounds for similar results of other finite discrete dynamics systems (Boolean automata networks, for example). From the other hand, when suitably constrained, they are a simple model of parsimonious systems (\ie systems which have a description exponentially shorter than the complete description of their dynamics). A natural research direction consists in trying to understand what exactly causes the complexity jumps between similar problems (cf. the last two entries of Table~\ref{tab:comprehensive}, for example).

A more perspective research direction consists in formalizing a general theory of the computational complexity of the dynamics of parsimonious discrete dynamical systems.
\begin{table}[t]
\begin{center}
\footnotesize
\begin{tabular}{*{2}{|l}|}
\textbf{Problem dealt with in this paper}&
\multicolumn{1}{l}{\textbf{Complexity}}
\\
\noalign{\global\dimen1\arrayrulewidth
\global\arrayrulewidth2pt}
\hline
\noalign{\global\arrayrulewidth\dimen1}
\strut
$\exists$ periodic point of period at most $\ell$ (*) &\NP-complete\\
\hline
$T$ has period $\leq\ell$ \& $\exists U$ reaching $T$ in  $\geq k$ steps (*)&\NP-complete\\
\hline
$\exists$ periodic point with period $\leq\ell$ \& preperiod len.  $\geq k$ (*)&\NP-complete\\
\hline
$T$ is global attractor of period $\leq\ell$ ($\geq\ell$) (*)&\PSPACE-complete\\
\hline
Reachability (*) &\PSPACE-complete\\
\hline
$\exists$ global fixed point attractor of period $\leq\ell$ ($\geq\ell$) (*)&\PSPACE-complete\\
\hline
$T$ belongs to a cycle of period (*) $\geq\ell$&\PSPACE-complete\\
\hline
$\res$ is bijective (=)&\coNP-complete\\
\hline
$\exists$ a periodic point with period $\leq\ell$ \& preperiod len.  $\geq k$ (*)&\NP-complete\\
\hline
Reachability in $k$ steps &\P-complete\\
\hline
$T$ has period $\leq\ell$ &\P-complete\\
\hline
$T$ has period $\leq\ell$ \& preperiod $\geq k$&\P-complete\\
\hline
$T$ has a k-predecessor&\P-complete\\
\hline
$\res$ is the identity&\coNP-complete\\
\hline
$T$ is an attractor cycle of period $\leq\ell$ \& diameter $\leq d$&\coNP-complete\\
\hline
$T$ is an attractor cycle of period $\leq\ell$ \& diameter $\geq d$&\NP-complete\\
\hline
$\exists$ an attractor cycle of period $\leq\ell$ \& diameter $\geq d$&\NP-complete\\
\hline
$\exists$ an attractor cycle of period $\leq\ell$ \& diameter $\leq d$&$\SigmaTwoP$-complete\\
\noalign{\global\dimen1\arrayrulewidth
\global\arrayrulewidth1pt}
\hline
\noalign{\global\arrayrulewidth\dimen1}
\end{tabular}
\end{center}
\caption{Comprehensive list of problems on the dynamics of RA and their complexity. (*)/(=): extended/same version of a result presented at either CIE 2014~\cite{Formenti2014a} or DCFS 2014~\cite{Formenti2014b}.}
\label{tab:comprehensive}
\end{table}
\section*{Acknowledgements}

Enrico Formenti acknowledges the partial support from the project PACA APEX FRI.
Alberto Dennunzio and Luca Manzoni were partially supported by Fondo d'Ateneo (FA)  2016 of Università degli Studi di Milano Bicocca: ``Sistemi Complessi e Incerti: teoria ed applicazioni''.
Antonio E. Porreca was partially supported by Fondo d'Ateneo (FA) 2015 of Università degli Studi di Milano Bicocca: ``Complessità computazionale e applicazioni crittografiche di modelli di calcolo bioispirati''.

\bibliographystyle{model1b-num-names}
\bibliography{Bibliography}

\end{document}